\documentclass[submission,copyright,creativecommons]{eptcs}
\providecommand{\event}{FICS 2013} 
\usepackage{breakurl}             

\usepackage{makeidx}
\usepackage{amssymb,amsmath,amsthm}
\usepackage{bbm}
\usepackage[curve,cmtip]{xypic}
\usepackage{color}
\theoremstyle{plain}
\newtheorem{theorem}{Theorem}[section]

\newtheorem{proposition}[theorem]{Proposition}
\newtheorem{definition}[theorem]{Definition}

\theoremstyle{remark}
\newtheorem{remark}{Remark}[section]
\newtheorem{example}[remark]{Example}

\renewcommand{\Im}{{\mathsf{Im}}}

\newcommand{\TRel}{{\mathsf{Rel}}}
\newcommand{\SRel}{{\mathsf{Rel}}}

\newcommand{\Rel}{{\mathsf{Rel}}}
\newcommand{\op}{{\mathsf{op}}}

\newcommand{\T}{{\mathsf{T}}}

\newcommand{\C}{{\mathbb{C}}}
\newcommand{\Set}{{\mathsf{Set}}}
\newcommand{\PS}{{\mathsf{P}}}

\newcommand{\Id}{{\mathsf{Id}}}

\newcommand{\Eq}{{\mathsf{Eq}}}
\newcommand{\Op}{{\mathcal{O}}}

\newcommand{\Pow}{{\mathcal{P}}}

\newcommand{\TW}{{\T}_W}
\newcommand{\SDist}{{\mathcal{S}}}

\newcommand{\tw}{{\mathsf{tw}}}
\newcommand{\tr}{{\mathsf{tr}}}
\newcommand{\ftr}{{\mathsf{ftr}}}

\newcommand{\st}{{\mathsf{st}}}
\newcommand{\dst}{{\mathsf{dst}}}
\renewcommand{\sup}{{\mathsf{sup}}}

\title{From Branching to Linear Time, Coalgebraically}
\author{Corina C\^{\i}rstea
\institute{University of Southampton}
\email{cc2@ecs.soton.ac.uk}}
\def\titlerunning{From Branching to Linear Time, Coalgebraically}
\def\authorrunning{C. C\^{\i}rstea}
\begin{document}
\maketitle
\begin{abstract}
We consider state-based systems modelled as coalgebras whose type incorporates branching, and show that by suitably adapting the definition of coalgebraic bisimulation, one obtains a general and uniform account of the linear-time behaviour of a state in such a coalgebra. By moving away from a boolean universe of truth values, our approach can measure the extent to which a state in a system with branching is able to exhibit a particular linear-time behaviour. This instantiates to measuring the probability of a specific behaviour occurring in a probabilistic system, or measuring the minimal cost of exhibiting a specific behaviour in the case of weighted computations.
\end{abstract}

\section{Introduction}
\label{1}

When analysing process behaviour, one of the early choices one has to make is between a linear and a branching view of time. In branching-time semantics, the choices a process has for proceeding from a particular state are taken into account when defining a notion of process equivalence (with bisimulation being the typical such equivalence), whereas in linear-time semantics such choices are abstracted away and the emphasis is on the individual executions that a process is able to exhibit. From a system verification perspective, one often chooses the linear-time view, as this not only leads to simpler specification logics and associated verification techniques, but also meets the practical need to verify all possible system executions.

While the theory of coalgebras has, from the outset, been able to provide a uniform account of various bisimulation-like observational equivalences (and later, of various simulation-like behavioural preorders), it has so far not been equally successful in giving a generic account of the linear-time behaviour of a state in a system whose type incorporates a notion of branching. For example, the generic trace theory of \cite{HasuoJS07} only applies to systems modelled as coalgebras of type $\T \circ F$, with the monad $\T : \Set \to \Set$ specifying a branching type (e.g.~non-deterministic or probabilistic), and the endofunctor $F : \Set \to \Set$ defining the structure of individual transitions (e.g.~labelled transitions or successful termination). The approach in loc.\,cit.~is complemented by that of \cite{JacobsSS12}, where traces are derived using a determinisation procedure similar to the one for non-deterministic automata. The latter approach applies to systems modelled as coalgebras of type $G \circ \T$, where again a monad $\T : \Set \to \Set$ is used to model branching behaviour, and an endofunctor $G$ specifies the transition structure. Neither of these approaches is able to account for potentially infinite traces, as typically employed in model-based formal verification. This limitation is partly addressed in \cite{cirstea-11}, but again, this only applies to coalgebras of type $\T \circ F$, albeit with more flexibility in the underlying category (which in particular allows a measure-theoretic account of infinite traces in probabilistic systems). Finally, none of the above-mentioned approaches exploits the compositionality that is intrinsic to the coalgebraic approach. In particular, coalgebras of type $G \circ \T \circ F$ (of which systems with both inputs and outputs are an example, see Example~\ref{input-output}) can not be accounted for by any of the existing approaches. This paper presents an attempt to address the above limitations concerning the types of coalgebras and the nature of traces that can be accounted for, by providing a \emph{uniform} and \emph{compositional} treatment of (possibly infinite) linear-time behaviour in systems with branching.

In our view, one of the reasons for only a partial success in developing a fully general coalgebraic theory of traces is the long-term aspiration within the coalgebra community to obtain a uniform characterisation of trace equivalence via a finality argument, in much the same way as is done for bisimulation (in the presence of a final coalgebra). This encountered difficulties, as a suitable category for carrying out such an argument proved difficult to find in the general case. In this paper, we tackle the problem of getting a handle on the linear-time behaviour of a state in a coalgebra with branching from a different angle: we do not attempt to directly define a notion of trace equivalence between two states (e.g.~via finality in some category), but focus on \emph{testing} whether a state is able to exhibit a particular trace, and on  measuring the extent of this ability. This "measuring" relates to the type of branching present in the system, and instantiates to familiar concepts such as the probability of exhibiting a given trace in probabilistic systems, the minimal cost of exhibiting a given trace in weighted computations, and simply the ability to exhibit a trace in non-deterministic systems.

The technical tool for achieving this goal is a generalisation of the notions of relation and relation lifting \cite{HermidaJ98}, which lie at the heart of the definition of coalgebraic bisimulation. Specifically, we employ relations valued in a partial semiring, and a corresponding generalised version of relation lifting. Our approach applies to coalgebras whose type is obtained as the composition of several endofunctors on $\Set$: one of these is a monad $\T$ that accounts for the presence of branching in the system, while the remaining endofunctors, assumed here to be polynomial, jointly determine the notion of linear-time behaviour. This strictly subsumes the types of systems considered in earlier work on coalgebraic traces \cite{HasuoJS07,cirstea-11,JacobsSS12}, while also providing compositionality in the system type.

Our main contribution, presented in Section~\ref{linear-time}, is a \emph{uniform} and \emph{compositional} account of linear-time behaviour in state-based systems with branching. A by-product of our work is an extension of the study of additive monads carried out in \cite{Kock2011,CoumansJ2011} to what we call \emph{partially additive monads} (Section~\ref{semiring}). Our approach can be summarised as follows:
\begin{itemize}
\item We move from two-valued to multi-valued relations, with the universe of truth values being induced by the choice of monad for modelling branching. This instantiates to relations valued in the interval $[0,1]$ in the case of probabilistic branching, the set $\mathbb N^\infty = \mathbb N \cup \{\infty\}$ in the case of weighted computations, and simply $\{\bot,\top\}$ in the case of non-deterministic branching. This reflects our view that the notion of truth used to reason about the observable behaviour of a system should be dependent on the branching behaviour present in that system. Such a dependency is also expected to result in temporal logics that are more natural and more expressive, and at the same time have a conceptually simpler semantics. In deriving a suitable structure on the universe of truth values, we generalise results on additive monads \cite{Kock2011,CoumansJ2011} to \emph{partially additive monads}. This allows us to incorporate probabilistic branching under our approach. We show that for a commutative, partially additive monad $\T$ on $\Set$, the set $\T 1$ carries a partial semiring structure with an induced preorder, which in turn makes $\T 1$ an appropriate choice of universe of truth values.
\item We generalise and adapt the notion of relation lifting used in the definition of coalgebraic bisimulation, in order to (i) support multi-valued relations, and (ii) abstract away branching. Specifically, we make use of the partial semiring structure carried by the universe of truth values to generalise relation lifting of polynomial endofunctors to multi-valued relations, and employ a canonical \emph{extension lifting} induced by the monad $\T$ to capture a move from branching to linear time. The use of this extension lifting allows us to make formal the idea of testing whether, and to what extent, a state in a coalgebra with branching can exhibit a particular \emph{linear-time} behaviour. Our approach resembles the idea employed by partition refinement algorithms for computing bisimulation on labelled transition systems with finite state spaces \cite{KS90}. There, one starts from a single partition of the state space, with all states related to each other, and repeatedly refines it through stepwise unfolding of the transition structure, until a fixpoint is reached. Similarly, we start by assuming that a state in a system with branching can exhibit any linear-time behaviour, and moreover, assign the maximum possible value to each pair consisting of a state and a linear-time behaviour. We then repeatedly refine the values associated to such pairs, through stepwise unfolding of the coalgebraic structure.
\end{itemize}

The present work is closely related to our earlier work on maximal traces and path-based logics \cite{cirstea-11}, which described a game-theoretic approach to testing if a system with non-deterministic branching is able to exhibit a particular trace. Here we consider arbitrary branching types, and while we do not emphasise the game-theoretic aspect, our use of greatest fixpoints has a very similar thrust.

\paragraph{Acknowledgements} Several fruitful discussions with participants at the 2012 Dagstuhl Seminar on Coalgebraic Logics helped refine the ideas presented here. Our use of relation lifting was inspired by the recent work on coinductive predicates \cite{Hasuo12}, itself based on the seminal work in \cite{HermidaJ98} on the use of predicate and relation lifting in the formalisation of induction and coinduction principles. Last but not least, the comments received from the anonymous reviewers contributed to improving the presentation of this work and to identifying new directions for future work.

\section{Preliminaries}

\subsection{Relation Lifting}
\label{rel-lifting}
The concepts of \emph{predicate lifting} and \emph{relation lifting}, to our knowledge first introduced in \cite{HermidaJ98}, are by now standard tools in the study of coalgebraic models, used e.g.~to provide an alternative definition of the notion of bisimulation (see e.g.~in \cite{JacobsBook}), or to describe the semantics of coalgebraic modal logics \cite{Pattinson03,Moss99}. While these concepts are very general, their use so far usually restricts this generality by viewing both predicates and relations as sub-objects in some category (possibly carrying additional structure). In this paper, we make use of the full generality of these concepts, and move from the standard view of relations as subsets to a setting where relations are valuations into a universe of truth values. This section recalls the definition of relation lifting in the standard setting where relations are given by monomorphic spans.

Throughout this section (only), {$\Rel$} denotes the category whose objects are binary relations $(R,\langle r_1,r_2 \rangle)$ with $\langle r_1,r_2 \rangle : R \to X \times Y$ a monomorphic span, and whose arrows from $(R,\langle r_1,r_2 \rangle)$ to $(R',\langle r_1',r_2' \rangle)$ are given by pairs of functions $(f : X \to X'\,,\, g : Y \to Y')$ ~s.t.~ $(f \times g) \circ  \langle r_1,r_2 \rangle$ factors through $\langle r_1',r_2' \rangle$:
\[\UseComputerModernTips\xymatrix{
R \ar@{-->}[d] \ar@{>->}[r]^-{\langle r_1,r_2 \rangle} & X \times Y \ar[d]^-{f \times g} \\
R' \ar@{>->}[r]^-{\langle r_1',r_2' \rangle} & X' \times Y' }\]
In this setting, the \emph{relation lifting of a functor $F : \Set \to \Set$} is defined as a functor $\Rel(F) : \Rel \to \Rel$ taking a relation $\langle r_1,r_2 \rangle : R \to X \times Y$ to the relation defined by the span $\langle F(r_1),F(r_2) \rangle : F(R) \to F(X) \times F(Y)$, obtained via the unique epi-mono factorisation of $\langle F(r_1),F(r_2) \rangle$:
\[\UseComputerModernTips\xymatrix{
R \ar@{>->}[d]_-{\langle r_1,r_2 \rangle} & F(R) \ar[d]_-{\langle F(r_1),F(r_2) \rangle} \ar@{->>}[r] & \Rel(F)(R)\ar@{>->} [dl]\\
X \times Y & F(X) \times F(Y)
}\]
It follows easily that this construction is functorial, and in particular preserves the order $\le$ between relations on the same objects given by $(R,\langle r_1,r_2 \rangle) \le (S,\langle s_1,s_2 \rangle)$ if and only if $\langle r_1,r_2 \rangle$ factors through $\langle s_1,s_2 \rangle$:
\[\UseComputerModernTips\xymatrix@+0.3pc{
R \ar@{>-->}[r] \ar@{>->}@<-1.5ex>[rr]_-{\langle r_1,r_2 \rangle} & S \ar@{>->}[r]^-{\langle s_1,s_2 \rangle} & X \times Y}\]
An alternative definition of $\Rel(F)$ for $F$ a \emph{polynomial functor} (i.e.~constructed from the identity and constant functors using \emph{finite} products and set-indexed coproducts) can be given by induction on the structure of $F$. We refer the reader to \cite[Section~3.1]{JacobsBook} for details of this definition. An extension of this definition to a more general notion of relation will be given in Section~\ref{gen-rel-lifting}.

\subsection{Coalgebras}
\label{section-coalgebras}
We model state-based, dynamical systems as coalgebras over the category of sets. Given a functor $F : \C \to \C$ on an arbitrary category, an \emph{$F$-coalgebra} is given by a pair $(C,\gamma)$ with $C$ an object of $\C$, used to model the state space, and $\gamma : C \to F C$ a morphism in $\C$, describing the one-step evolution of the system states. Then, a canonical notion of observational equivalence between the states of two $F$-coalgebras is provided by the notion of bisimulation. Of the many, and under the assumption that $F$ preserves weak pullbacks, equivalent definitions of bisimulation (see \cite{JacobsBook} for a detailed account), we recall the one based on relation lifting. This applies to coalgebras over the category of sets (as described below), but also more generally to categories with logical factorisation systems (as described in \cite{JacobsBook}). According to this definition, an \emph{$F$-bisimulation} between coalgebras $(C,\gamma)$ and $(D,\delta)$ over $\Set$ is a $\Rel(F)$-coalgebra:
\[\UseComputerModernTips\xymatrix{
R \ar@{-->}[r] \ar@{>->}[d] & \Rel(F)(R)\ar@{>->}[d]\\
X \times Y \ar[r]_-{\gamma \times \delta} & F(X) \times F(Y)
}\]
In the remainder of this section we sketch a coalgebraic generalisation of a well-known partition refinement algorithm for computing \emph{bisimilarity} (i.e.~the largest bisimulation) on finite-state labelled transition systems \cite{KS90}. For an arbitrary endofunctor $F : \Set \to \Set$ and two finite-state $F$-coalgebras $(C,\gamma)$ and $(D,\delta)$, the generalised algorithm iteratively computes relations $\simeq_i \,\,\subseteq\,\, C \times D$ with $i = 0,1, \ldots$ as follows:
\begin{itemize}
\item $\sim_0 \,=\, C \times D$
\item $\sim_{i+1} \,=\, (\gamma \times \delta)^*(\Rel(F)(\simeq_i))$ for $i = 0,1,\ldots$
\end{itemize}
where $(\gamma \times \delta)^*$ takes a relation $R \subseteq F C \times F D$ to the relation $\{(c,d) \in C \times D \mid (\gamma(c),\delta(d)) \in R \}$. Thus, in the initial approximation $\simeq_0$ of the bisimilarity relation, all states are related, whereas at step $i+1$ two states are related if and only if their one-step observations are suitably related using the relation $\simeq_i$. Bisimilarity between the coalgebras $(C,\gamma)$ and $(D,\delta)$ thus arises as the greatest fixpoint of a monotone operator on the complete lattice of relations between $C$ and $D$, which takes a relation $R \subseteq C \times D$ to the relation $(\gamma \times \delta)^*(\Rel(F)(R))$. A similar characterisation of bisimilarity exists for coalgebras with infinite state spaces, but in this case the fixpoint can not, in general, be reached in a finite number of steps.

The above greatest fixpoint characterisation of bisimilarity is generalised and adapted in Section~\ref{linear-time}, in order to characterise the extent to which a state in a coalgebra with branching can exhibit a linear-time behaviour. There, the two coalgebras in question have different types: the former has branching behaviour and is used to model the system of interest, whereas the latter has linear behaviour only and describes the domain of possible traces.

\subsection{Monads}

In what follows, we use monads $(\T,\eta,\mu)$ on $\Set$ (where $\eta : \Id \Rightarrow \T$ and $\mu : \T \circ \T \Rightarrow \T$ are the \emph{unit} and \emph{multiplication} of $\T$) to capture branching in coalgebraic types. Moreover, we assume that these monads are \emph{strong} and \emph{commutative}, i.e.~they come equipped with a \emph{strength map} $\st_{X,Y} : X \times \T Y \to \T(X \times Y)$ as well as a \emph{double strength map} $\dst_{X,Y} : \T X \times \T Y \to \T(X \times Y)$ for each choice of sets $X,Y$; these maps are natural in $X$ and $Y$, and satisfy coherence conditions w.r.t.~the unit and multiplication of $\T$. We also make direct use of the \emph{swapped strength map} $\st'_{X,Y} : \T X \times Y \to \T(X \times Y)$, obtained from the strength via the \emph{twist map} $\tw_{X,Y} : X \times Y \to Y \times X$:
\[\UseComputerModernTips\xymatrix@+0.3pc{\T X \times Y \ar[r]^-{\tw_{\T X,Y}} & Y \times \T X \ar[r]^-{\st_{Y,X}} & \T(Y \times X) \ar[r]^-{\T \tw_{Y,X}} & \T(X \times Y)}\]
\begin{example}
\label{example-monads}
As examples of monads, we consider:
\begin{enumerate}
\item the \emph{powerset monad} $\Pow : \Set \to \Set$, modelling nondeterministic computations, with unit given by singletons and multiplication given by unions. Its strength and double strength are given by 
\begin{align*}
\st_{X,Y}(x,V) = \{x\} \times V & & 
\dst_{X,Y}(U,V) = U \times V
\end{align*}
for $x\in X$, $U \in \Pow X$ and $V \in \Pow Y$,
\item the \emph{semiring monad} $\T_S : \Set \to \Set$ with $(S,+,0,\bullet,1)$ a semiring, given by
\[\T_S(X) = \{ f : X \to S \mid \sup(f) \text{ is finite} \}\]
with $\sup(f) = \{ x \in X \mid f(x) \ne 0\}$ the \emph{support} of $f$. Its unit and multiplication are given by
\begin{align*}
\eta_X(x)(y) = \begin{cases}1 & \text{ if } y = x \\
0 & \text{ otherwise} \end{cases} & & \mu_X(f \in S^{(S^X)}) = \sum\limits_{g \in \sup(f)}\sum\limits_{x\in \sup(g)} f(g) \bullet g(x)
\end{align*}
while its strength and double strength are given by
\begin{align*}
\st_{X,Y}(x,g)(z,y) = \begin{cases} g(y) & \text{ if } z = x \\
0 & \text{ otherwise}
\end{cases} & & \dst_{X,Y}(f,g)(z,y) = f(z) \bullet g(y)
\end{align*}
for $x \in X$, $f\in \T_S(X)$, $g \in \T_S(Y)$, $z \in X$ and $y \in Y$.
As a concrete example, we will consider the semiring $W = (\mathbb N^\infty,\min,\infty,+,0)$, and use $\T_W$ to model weighted computations.
\item the \emph{sub-probability distribution monad} $\SDist : \Set \to \Set$, modelling probabilistic computations, with unit given by the Dirac distributions (i.e.~$\eta_{X}(x) = (x \mapsto 1)$), and multiplication given by $\mu_X(\Phi) = \sum\limits_{\varphi \in \sup(\Phi)}\sum\limits_{x \in \sup(\varphi)} \Phi(\varphi) * \varphi(x)$, with $*$ denoting multiplication on $[0,1]$. Its strength and double strength are given by
\begin{align*}
\st_{X,Y}(x,\psi)(z,y) = \begin{cases}\psi(y) & \text{if } z = x\\ 0 & \text{otherwise}
\end{cases}  & & 
\dst_{X,Y}(\varphi,\psi)(z,y) = \varphi(z) * \psi(y)
\end{align*}
for $x \in X$, $\varphi \in \SDist(X)$, $\psi \in \SDist(Y)$, $z \in X$ and $y \in Y$.
\end{enumerate}
\end{example}

\section{From Partially Additive, Commutative Monads to Partial Commutative Semirings with Order}
\label{semiring}

Later in this paper we will consider coalgebras whose type is given by the composition of several endofunctors on $\Set$, one of which is a commutative monad $\T : \Set \to \Set$ accounting for the presence of branching in the systems of interest. This section extends results in \cite{Kock2011,CoumansJ2011} to show how to derive a universe of truth values from such a monad. The assumption of loc.\,cit.~concerning the \emph{additivity} of the monad under consideration is here weakened to \emph{partial additivity} (see Definition~\ref{additive}); this allows us to incorporate the sub-probability distribution monad (which is not additive) into our framework. Specifically, we show that any commutative, partially additive monad $\T : \Set \to \Set$ induces a partial commutative semiring structure on the set $\T 1$, with $1=\{*\}$ a final object in $\Set$. We recall that a \emph{commutative semiring} consists of a set $S$ carrying two commutative monoid structures $(+,0)$ and $(\bullet, 1)$, with the latter distributing over the former: $s \bullet 0 = 0$ and $s \bullet (t + u) = s \bullet t + s \bullet u$ for all $s,t,u \in S$. A \emph{partial commutative semiring} is defined similarly, except that $+$ is a partial operation subject to the condition that whenever $t + u$ is defined, so is $s \bullet t + s \bullet u$, and moreover $s \bullet (t + u) = s \bullet t + s \bullet u$. The relevance of a partial commutative semiring structure on the set of truth values will become clear in Sections~\ref{gen-rel-lifting} and \ref{linear-time}.

It follows from results in \cite{CoumansJ2011} that any commutative monad $(\T,\eta,\mu)$ on $\Set$ induces a commutative monoid $(\T(1),\bullet,\eta_1(*))$, with multiplication $\bullet : \T(1) \times \T(1) \to \T(1)$ given by the composition
\[\UseComputerModernTips\xymatrix@+1pc{\T (1) \times \T (1) \ar[r]^-{\dst_{1,1}} & \T(1 \times 1) \ar[r]^-{\T \pi_2} & \T (1)}\]
Alternatively, this multiplication can be defined as the composition
\[\UseComputerModernTips\xymatrix{\T(1) \times \T(1) \ar[r]^-{\st'_{1,1}} & \T(1 \times \T(1)) \ar[r]^-{T \pi_2} & \T^2(1) \ar[r]^-{\mu_1} & \T(1)}\]
or as
\[\UseComputerModernTips\xymatrix{\T(1) \times \T(1) \ar[r]^-{\st_{1,1}} & \T(\T(1) \times 1) \ar[r]^-{T \pi_1} & \T^2(1) \ar[r]^-{\mu_1} & \T(1)}\]
(While the previous two definitions coincide for commutative monads, this is not the case in general.)
\begin{remark}
\label{actions}
The following maps define left and right actions of $(\T(1),\bullet)$ on $\T( X)$:
\[\UseComputerModernTips\xymatrix{
\T(1) \times \T(X) \ar[r]^-{\dst_{1,X}} & \T(1 \times X) \ar[r]^-{\T \pi_2} & \T (X)
} \qquad \qquad \UseComputerModernTips\xymatrix{
\T(X) \times \T(1) \ar[r]^-{\dst_{X,1}} & \T(X \times 1) \ar[r]^-{\T \pi_1} & \T (X)
}\]
\end{remark}

On the other hand, any monad $\T :  \Set \to \Set$ with {$ \T \emptyset = 1$} is such that, for any $X$, $\T X$ has a \emph{zero element} $0 \in \T X$, obtained as $(\T!_X)(*)$. This yields a \emph{zero map $0 : Y \to \T X$} for any $X,Y$, obtained as the composition
\[\UseComputerModernTips\xymatrix{Y \ar[r]^-{!_Y} &  \T \emptyset \ar[r]^-{T !_X} & \T X}\]
with the maps $!_Y : Y \to \T \emptyset$ and $!_X : \emptyset \to X$ arising by finality and initiality, respectively. Now consider the following map:
\begin{equation}
\label{map}
\UseComputerModernTips\xymatrix@+4pc{T(X+Y) \ar[r]^-{\langle \mu_X \circ \T p_1,\mu_Y \circ \T p_2 \rangle} & \T X \times \T Y}
\end{equation}
where $p_1 = [\eta_X,0] : X + Y \to \T X$ and $p_2 = [0,\eta_Y] : X + Y \to \T Y$.
\begin{definition}
\label{additive}
A monad $\T : \Set \to \Set$ is called \emph{additive}\footnote{Additive monads were studied in \cite{Kock2011,CoumansJ2011}.} (\emph{partially additive}) if\, $\T \emptyset = 1$ and the map in (\ref{map}) is an isomorphism (respectively monomorphism).
\end{definition}
 The (partial) inverse of the map $\langle \mu_X \circ \T p_1,\mu_Y \circ \T p_2 \rangle$ can be used to define a (partial) addition on the set $\T X$, given by $\T[1_X,1_X] \circ q_{X,X}$, where $q_{X,X} : \T X \times \T X \to \T(X+X)$ is the (partial) left inverse of $\langle \mu_X \circ \T p_1,\mu_Y \circ \T p_2 \rangle$:
 \[\UseComputerModernTips\xymatrix@+1pc{\T X & \T(X+X) \ar@<+1ex>[rr]^-{\langle \mu_X \circ \T p_1,\mu_Y \circ \T p_2 \rangle} \ar[l]_-{\T[1_X,1_X]} & & \T X \times \T X \ar@<+1ex>@{-->}[ll]^-{q_{X,X}} \ar@<+1ex>@/^3ex/[lll]^-{{+}}}\]
That is, $a + b$ is defined if and only if $(a,b) \in \Im(\langle \mu_X \circ \T p_1,\mu_Y \circ \T p_2 \rangle)$\,\footnote{A similar, but \emph{total}, addition operation is defined in \cite{Kock2011,CoumansJ2011} for additive monads.}.

\cite[Section~5.2]{CoumansJ2011} explores the connection between additive, commutative monads and commutative semirings. The next result provides a generalisation to partially additive, commutative monads and partial commutative semirings. 

The proof of Proposition~\ref{prop-partial} is a slight adaptation of the corresponding proofs in \cite[Section~5.2]{CoumansJ2011}.
\begin{proposition}
\label{prop-partial}
Let $\T$ be a commutative, (partially) additive monad. Then:
\begin{enumerate}
\item $(\T 1,\bullet,\eta_1(*))$ is a commutative monoid.
\item $(\T X,0,+)$ is a (partial) commutative monoid, for each set $X$.
\item\label{3}  $(\T 1,0,+,\bullet,\eta_1(*))$ is a (partial) commutative semiring. 
\end{enumerate}
\end{proposition}
\begin{proof}[Proof (Sketch)]
The commutativity of the following diagram lies at the heart of the proof of item \ref{3}:
 \[\UseComputerModernTips\xymatrix{
 \T 1 \times \T 1 \ar[dd]_-{\bullet} & & \T (1+1) \times \T 1 \ar[ll]_-{\T[1_X,1_X] \times 1_{\T 1}} \ar[dd]_-{a_{\T(1+1)}} \ar@<+1ex>[r]^-{\delta \times 1_{\T 1} }& (\T 1 \times \T 1) \times \T 1 \ar@<+1ex>@{-->}[l]^-{q_{1,1} \times 1_{\T 1}} \ar[d]^-{\langle \pi_1 \times \pi_2,\pi_2 \times \pi_2 \rangle} \\
 & & & (\T 1 \times \T 1) \times (\T 1 \times \T 1) \ar[d]^-{\bullet \,\times\, \bullet}\\
 \T 1 & & \T(1 + 1) \ar@<+1ex>[r]^-{\delta} \ar[ll]^-{\T[1_X,1_X]} & \T 1 \times \T 1\ar@<+1ex>@{-->}[l]^-{q_{1,1}}
 }\]
 where $a_{\T X} : \T X \times \T 1 \to \T X$ is the right action from Remark~\ref{actions}, and $\delta$ is the map $\langle \mu_1 \circ \T p_1,\mu_1 \circ \T p_2 \rangle$ used in the definition of $+$ on $\T 1$. The composition $\bullet \circ (\T[1_X,1_X] \times 1_{\T 1}) \circ (q_{1,1} \times 1_{\T 1})$ captures the computation of $(a + b) \bullet c$, whereas the composition $\T[1_X,1_X] \circ q_{1,1} \circ (\bullet \times \bullet) \circ \langle \pi_1 \times \pi_2,\pi_2 \times \pi_2 \rangle$ captures the computation $a \bullet c + b \bullet c$, with $a,b,c \in \T 1$. The fact that $\delta$ commutes with the strength map (by (iv) of \cite[Lemma~15]{CoumansJ2011}), together with $a_{\T(1+1)}$ and $\bullet$ being essentially given by the double strength maps $\dst_{1+1,1}$ and $\dst_{1,1}$, yields $(\bullet \times \bullet) \circ \langle \pi_1 \times \pi_2,\pi_2 \times \pi_2 \rangle \circ (\delta \times 1_{\T 1}) = \delta \circ a_{\T(1+1)}$, that is, commutativity (via the plain arrows) of the right side of the above diagram. This immediately results in $a \bullet c + b \bullet c$ being defined whenever $a + b$ is defined, and hence in the commutativity of the right side of the diagram also via the dashed arrows. This, combined with the commutativity of the left side of the diagram (which is simply naturality of the right action $a$), gives  $(a + b) \bullet c = a \bullet c + b \bullet c$ whenever $a+b$ is defined.
\end{proof}
\begin{example}
\label{example-semirings}
For the monads in Example~\ref{example-monads}, one obtains the commutative semirings $(\{\bot,\top\},\vee,\bot,\wedge,\top)$ when {$\T = \Pow$}, $({\mathbb N}^\infty,\min,\infty,+,0)$ when {$\T = \TW$}\,\footnote{This is sometimes called the \emph{tropical semiring}.}, and the partial commutative semiring $([0,1],+,0,*,1)$ when $\T = \SDist$ (where in the latter case $a + b$ is defined if and only if $a + b \le 1$).
\end{example}

\section{Generalised Relations and Relation Lifting}
\label{gen-rel-lifting}

This section introduces generalised relations valued in a partial commutative semiring, and shows how to lift polynomial endofunctors on $\Set$ to the category of generalised relations. We begin by fixing a partial commutative semiring $(S,+,0,\bullet,1)$, and noting that the partial monoid $(S,+,0)$ can be used to define a preorder relation on $S$ as follows:
\[x \sqsubseteq y ~~~\text{if and only if}~~~ \text{there exists } z \in S \text{ such that }x + z = y\]
for $x,y \in S$. It is then straightforward to show (using the definition of a partial commutative semiring) that the preorder $\sqsubseteq$ has $0 \in S$ as bottom element, and is preserved by $\bullet$ in each argument. Proper (i.e.~not partial) semirings where the preorder $\sqsubseteq$ is a partial order are called \emph{naturally ordered} \cite{EsikK07}. We here extend this terminology to partial semirings.

\begin{example}
\label{example-orders}
For the monads in Example~\ref{example-monads}, the preorders associated to the induced partial semirings (see Example~\ref{example-semirings}) are all partial orders: $\le$ on $\{\bot,\top\}$ for $\T = {\Pow}$, $\le$ on $[0,1]$ for $\T = {\SDist}$, and $\ge$ on ${\mathbb N}^\infty$ for $\T = \TW$.
\end{example}

We let $\SRel$ denote the category\footnote{To keep notation simple, the dependency on $S$ is left implicit.} with objects given by triples $(X,Y,R)$, where $R : X \times Y \to S$ is a function defining a \emph{multi-valued relation} (or \emph{$S$-relation}), and with arrows from $(X,Y,R)$ to $(X',Y',R')$ given by pairs of functions $(f,g)$ as below, such that $R \sqsubseteq R' \circ (f \times g)$:
\[\UseComputerModernTips\xymatrix{X \times Y \ar@{}[dr]|-{\sqsubseteq}\ar[r]^-{f \times g} \ar[d]_-{R} & X' \times Y' \ar[d]^-{R'}\\ S \ar@{=}[r] & S}\]
Here, the order $\sqsubseteq$ on $S$ has been extended pointwise to $S$-relations with the same carrier.

We write $\Rel_{X,Y}$ for the \emph{fibre over $(X,Y)$}, that is, the full subcategory of $\Rel$ whose objects are $S$-relations over $X \times Y$ and whose arrows are given by $(1_X,1_Y)$. It is straightforward to check that the functor $q : \Rel \to \Set \times \Set$ taking $(X,Y,R)$ to $(X,Y)$ defines a fibration: the reindexing functor $(f,g)^* : \Rel_{X',Y'} \to \Rel_{X,Y}$ takes $R' : X' \times Y' \to S$ to $R' \circ (f \times g) : X \times Y \to S$. 

We now proceed to generalising relation lifting to $S$-relations.

\begin{definition}
\label{def-gen-rel-lifting}
Let $F : \Set \to \Set$. A \emph{relation lifting of $F$} is a functor\footnote{Given the definition of the fibration $q$, such a functor is automatically a morphism of fibrations.} $\Gamma : \SRel \to \SRel$ such that $q \circ \Gamma = (F \times F) \circ q$:
\[\UseComputerModernTips\xymatrix{
\SRel \ar[d]_-{q} \ar[r]^-{\Gamma} & \SRel \ar[d]^-{q} \\
\Set \times \Set \ar[r]_-{F \times F} & \Set \times \Set}\]
\end{definition}
We immediately note a fundamental difference compared to standard relation lifting as defined in Section~\ref{rel-lifting}. While in the case of standard relations each functor admits exactly one lifting, Definition~\ref{def-gen-rel-lifting} implies neither the existence nor the uniqueness of a lifting. We defer the study of a canonical lifting (similar to $\Rel(F)$ in the case of standard relations) to future work, and show how to define a relation lifting of $F$ in the case when $F$ is a polynomial functor. To this end, we make the additional assumption that the unit $1$ of the semiring multiplication is a top element (which we also write as $\top$) for the preorder $\sqsubseteq$. Recall that $\sqsubseteq$ also has a bottom element (which we will sometimes denote by $\bot$), given by the unit $0$ of the (partial) semiring addition. The definition of the relation lifting of a polynomial functor $F$ is by structural induction on $F$ and makes use of the semiring structure on $S$:
\begin{itemize}
\item If $F = \Id$, $\Rel(F)$ takes an $S$-relation to itself.
\item If $F = C$, $\Rel(F)$ takes an $S$-relation to the equality relation $\Eq(C) : C \times C \to S$ given by
\[\Eq_C(c,c') ~=~ \begin{cases} \top \text{ if }c = c' \\
\bot \text{ otherwise} \end{cases}\]
\item If $F = F_1 \times F_2$,  $\Rel(F)$ takes an $S$-relation $R : X \times Y \to S$ to:
\[\!\!\!\!\!\!\UseComputerModernTips\xymatrix@-1.1pc{(F_1 X \times F_2 X) \times (F_1 Y \times F_2 Y) \ar[rrr]^-{\langle \pi_1 \times \pi_1,\pi_2 \times \pi_2 \rangle} & & & (F_1 X \times F_1 Y) \times (F_2 X \times F_2 Y) \ar[rrrrr]^-{\Rel(F_1)(R) \times \Rel(F_2)(R)} & & & & & S \times S \ar[r]^-{\bullet} & S}\]
The functoriality of this definition follows from the preservation of $\sqsubseteq$ by $\bullet$ (see Section~\ref{semiring}).
\item if $F = F_1 + F_2$, $\Rel(F)(R) : (F_1 X + F_2 X) \times (F_1 Y + F_2 Y ) \to S$ is defined by case analysis:
\begin{align*}
\Rel(F)(R)(\iota_i(u),\iota_j(v)) & ~=~ \begin{cases} \Rel(F_i)(R)(u,v) & \text{ if } i = j\\
\bot & \text{ otherwise} \end{cases}
\end{align*}
for $i,j \in \{1,2\}$, $u \in F_i X$ and $v \in F_j Y$. This definition generalises straightforwardly from binary to set-indexed coproducts.
\end{itemize}

\begin{remark}
A more general definition of relation lifting, which applies to arbitrary functors on $\Set$, is outside the scope of this paper. We note in passing that such a relation lifting could be defined by starting from a \emph{generalised predicate lifting} $\delta : F \circ \PS_0 \Rightarrow \PS_0 \circ F$ for the functor $F$, similar to the predicate liftings used in the work on coalgebraic modal logic \cite{Pattinson03}. Here, the contravariant functor $\PS_0 : \Set \to \Set^\op$ takes a set $X$ to the hom-set $\Set(X,S)$. Future work will also investigate the relevance of the results in \cite{Ghani2011,Ghani2012} to a general definition of relation lifting in our setting. Specifically, the work in loc.\,cit.~shows how to construct truth-preserving predicate liftings and equality-preserving relation liftings for arbitrary functors on the base category of a \emph{Lawvere fibration}, to the total category of that fibration.
\end{remark}
For the remainder of this paper, we take $(S,+,0,\bullet,1)$ to be the partial semiring derived in Section~\ref{semiring} from a commutative, partially additive monad $\T$, and we view $S$ as the set of truth values.
In the case of the powerset monad, this corresponds to the standard view of relations as subsets, whereas in the case of the sub-probability distribution monad, this results in relations given by valuations in the interval $[0,1]$.

\begin{example}
Let $F: \Set \to \Set$ be given by $F X = 1 + A \times X$, with $A$ a set (of labels), and let $(S,+,0,\bullet,1)$ be the partial semiring with carrier $\T 1$ defined in Section~\ref{semiring}.
\begin{itemize}
\item For $\T = \Pow$, $\Rel(F)$ takes a (standard) relation $R \subseteq X \times Y$ to the relation
\[\{(\iota_1(*),\iota_1(*)\} \cup \{((a,x),(a,y)) \mid a \in A, (x,y) \in R \}\]
\item For $\T = \SDist$, $\Rel(F)$ takes $R : X \times Y \to [0,1]$ to the relation $R' : F X \times F Y \to [0,1]$ given by
\[ R'(\iota_1(*),\iota_1(*)) = 1 ~\qquad~
R'((a,x),(a,y)) = R(x,y) ~\qquad~
R'(u,v) = 0 ~\text{ in all other cases}
\]
\item For $\T = \TW$, $\Rel(F)$ takes $R : X \times Y \to \mathbb N^\infty$ to the relation $R' : F X \times F Y \to \mathbb N^\infty$ given by
\[
R'(\iota_1(*),\iota_1(*)) = 0 ~\qquad~
R'((a,x),(a,y)) = R(x,y) ~\qquad~
R'(u,v) = \infty ~\text{ in all other cases}
\]
\end{itemize}
\end{example}

\section{From Bisimulation to Traces}
\label{linear-time}

Throughout this section we fix a commutative, partially additive monad $\T : \Set \to \Set$ and assume, as in the previous section, that the natural preorder $\sqsubseteq$ induced by the partial commutative semiring obtained in Section~\ref{semiring} has the multiplication unit $\eta_1(*) \in \T 1$ as top element. Furthermore, we assume that this preorder is an \emph{$\omega^{\op}$-chain complete} partial order, where $\omega^{\op}$-chain completeness amounts to  any decreasing chain $x_1 \sqsupseteq x_2 \sqsupseteq \ldots$ having a greatest lower bound $\sqcap_{i \in \omega} x_i$. These assumptions are clearly satisfied by the orders in Example~\ref{example-orders}. 

We now show how combining the liftings of polynomial functors to the category of generalised relations valued in the partial semiring $\T 1$ (as defined in Section~\ref{gen-rel-lifting}) with so-called \emph{extension liftings} which arise canonically from the monad $\T$, can be used to give an account of the linear-time behaviour of a state in a coalgebra with branching. The type of such a coalgebra can be any composition involving polynomial endofunctors and the branching monad $\T$, although compositions of type $\T \circ F$, $G \circ \T$ and $G \circ \T \circ F$ with $F$ and $G$ polynomial endofunctors are particularly emphasised in what follows.

We begin with some informal motivation. When $\Rel$ is the standard category of binary relations, recall from Section~\ref{section-coalgebras} that an $F$-bisimulation is simply a $\Rel(F)$-coalgebra, and that the largest $F$-bisimulation between two $F$-coalgebras $(C,\gamma)$ and $(D,\delta)$ can be obtained as the greatest fixpoint of the monotone operator  on $\Rel_{C \times D}$ which takes a relation $R$ to the relation $(\gamma \times \delta)^*( \Rel(F)(R))$. Generalising the notion of $F$-bisimulation from standard relations to $\T 1$-relations makes little sense when the systems of interest are $F$-coalgebras. However, when considering say, coalgebras of type $\T \circ F$, it turns out that liftings of $F$ to the category of $\T 1$-relations (as defined in Section~\ref{gen-rel-lifting}) can be used to describe the \emph{linear-time behaviour} of states in such a coalgebra, when combined with suitable liftings of $\T$ to the same category of relations. To see why, let us consider labelled transition systems viewed as coalgebras of type $\Pow(1 + A \times \Id)$. In such a coalgebra $\gamma : C \to \Pow(1 + A \times C)$, explicit termination is modelled via transitions $c \to \iota_1(*)$, whereas deadlock (absence of a transition) is modelled as $\gamma(c) = \emptyset$. In this case, $\Rel(\Pow) \circ \Rel(1 + A \times \Id)$ is naturally isomorphic to $\Rel(\Pow(1 + A \times \Id))$\,\footnote{A similar observation holds more generally for $\Pow \circ F$ with $F$ a polynomial endofunctor. In general, only a natural transformation $\Rel(F \circ G) \Rightarrow \Rel(F) \circ \Rel(G)$ exists, see \cite[Exercise~4.4.6]{JacobsBook}.}, and takes a relation $R \subseteq X \times Y$ to the relation $R' \subseteq \Pow(1 + A \times X) \times \Pow(1 + A \times Y)$ given by
\[(U,V) \in R' ~~~\text{ if and only if }~~~ \begin{cases}
\text{if } \iota_1(*) \in U \text{ then } \iota_1(*) \in V , \text{ and conversely}\\
\text{if } (a,x) \in U \text{ then there exists } (a,y) \in V \text{ with } (x,y) \in R, \text{ and conversely}
\end{cases}
\]
Thus, the largest $\Pow(1 + A \times \Id)$-bisimulation between two coalgebras $(C,\gamma)$ and $(D,\delta)$ can be computed as the greatest fixpoint of the operator on $\Rel_{C,D}$ obtained as the composition
\begin{equation}
\label{el}
\UseComputerModernTips\xymatrix@+0.6pc{
R \subseteq C \times D \ar@{|->}[r]^-{\Rel(F)} & R_1 \subseteq F C \times F D \ar@{|->}[r]^-{\Rel(\Pow)} & R_2 \subseteq \Pow(F C) \times \Pow (F D) \ar@{|->}[r]^-{(\gamma \times \delta)^*} & R'\subseteq C \times D}
\end{equation}
where $F = 1 + A \times \Id$. Note first that $\Rel(\Pow)$ (defined in Section~\ref{rel-lifting} for an arbitrary endofunctor on $\Set$) takes a relation $R \subseteq X \times Y$ to the relation $R' \subseteq \Pow(X) \times \Pow(Y)$ given by
\[(U,V) \in R' \text{ ~if and only if~ for all } x \in U \text{ there exists } y \in V \text{ with } (x,y) \in R, \text{ and conversely}\]
Now consider the effect of replacing $\Rel(\Pow)$ in (\ref{el}) with the lifting $L : \Rel \to \Rel$ that takes a relation $R \subseteq X \times Y$ to the relation $R' \subseteq \Pow(X) \times Y$ given by
\[(U,y) \in R' \text{ ~if and only if~ there exists } x \in U \text{ with } (x,y) \in R\]
To do so, we must change the type of the coalgebra $(D,\delta)$ from $\Pow \circ F$ to just $F$. A closer look at the resulting operator on $\Rel_{C,D}$ reveals that it can be used to test for the existence of a matching trace: each state of the $F$-coalgebra $(D,\delta)$ can be associated a \emph{maximal trace}, i.e.~an element of the final $F$-coalgebra, by finality. In particular, when $F = 1 + A \times \Id$, maximal traces are either finite or infinite sequences of elements of $A$. Thus, the greatest fixpoint of the newly defined operator on $\Rel_{C \times D}$ corresponds to the relation on $C \times D$ given by
\begin{eqnarray*}c \ni_\tr d \text{ ~if and only if~  there exists a sequence of choices of transitions starting from } c \in C \text{ that leads to}\\
\qquad \qquad \quad \text{ exactly the same maximal trace (element of $A^* \cup A^\omega)$ as the single trace of } d \in D
\end{eqnarray*}
This relation models the ability of the state $c$ to exhibit the same trace as that of $d$.

The remainder of this section formalises the above intuitions, and generalises them to arbitrary monads $\T$ and polynomial endofunctors $F$, as well as to arbitrary compositions involving the monad $\T$ and polynomial endofunctors. We begin by restricting attention to coalgebras of type $\T \circ F$, with the monad $\T$ capturing branching and the endofunctor $F$ describing the structure of individual transitions. In this case it is natural to view the elements of the final $F$-coalgebra as possible \emph{linear-time} observable behaviours of states in $\T \circ F$-coalgebras. Similarly to the above discussion, we let $(C,\gamma)$ and $(D,\delta)$ denote a $\T \circ F$-coalgebra and respectively an $F$-coalgebra. The lifting of $F$ to $\T 1$-relations will be used as part of an operator on $\Rel_{C,D}$. In order to generalise the lifting $L$ above to arbitrary monads $\T$, we recall the following result from \cite{Kock12}, which assumes a strong monad $\T$ on a cartesian closed category.

\begin{proposition}[{\cite[Proposition~4.1]{Kock12}}]
\label{prop-kock}
Let $(B,\beta)$ be a $\T$-algebra. For any $f : X \times Y \to B$, there exists a unique $1$-linear $\overline{f} : \T X \times Y \to B$ making the following triangle commute:
\[\UseComputerModernTips\xymatrix{
\T X \times Y \ar[r]^-{\overline{f}} & B \\
X \times Y \ar[u]^-{\eta_X \times 1_Y} \ar[ur]_-{f}
}\]
\end{proposition}
In the above, \emph{$1$-linearity} is linearity in the first variable. More precisely, for $\T$-algebras $(A,\alpha)$ and $(B,\beta)$, a map $f : A \times Y \to B$ is called \emph{$1$-linear} if the following diagram commutes:
\[\UseComputerModernTips\xymatrix{
\T(A) \times Y \ar[r]^-{\st'_{A,Y}} \ar[d]_-{\alpha \times 1_Y} & \T(A \times Y) \ar[r]^-{\T(f)} & \T(B) \ar[d]^-{\beta}\\
A \times Y \ar[rr]_-{f} & & B
}\]
Clearly $1$-linearity should be expected of the lifting $L(R) : \T X \times Y \to \T1$ of a relation $R : X \times Y \to \T 1$, as this amounts to $L(R)$ commuting with the $\T$-algebra structures $(\T X,\mu_X)$ and $(\T 1,\mu_1)$. Given this, the diagram of Proposition~\ref{prop-kock} forces the definition of the generalised lifting.

\begin{definition}
The \emph{extension lifting} $L_\T : \Rel \to \Rel$ is the functor taking a relation $R : X \times Y \to \T1$ to its unique $1$-linear extension $\overline{R} : \T X \times Y \to \T1$.
\end{definition}

\begin{remark}
It follows from \cite{Kock12} that a direct definition of the relation $\overline{R} : \T X \times Y \to \T 1$ is as the composition
\[\UseComputerModernTips\xymatrix@+0.3pc{
\T X \times Y \ar[r]^-{\st'_{X,Y}} & \T(X \times Y) \ar[r]^-{\T (R)} & \T^2 1 \ar[r]^-{\mu_1} & \T 1
}\]
This also yields functoriality of $L_\T$, which follows from the functoriality of its restriction to each fibre category $\Rel_{X,Y}$, as proved next.
\end{remark}

\begin{proposition}
\label{prop-functoriality}
The mapping $R \in \Rel_{X,Y} \mapsto \overline{R} \in \Rel_{\T X,Y}$ is functorial.
\end{proposition}
\begin{proof}[Proof (Sketch)]
Let $R,R' \in \Rel_{X,Y}$ be such that $R \sqsubseteq R'$. Hence, there exists $S \in \Rel_{X,Y}$ such that $R + S = R'$ (pointwise). To show that $\overline{R} \sqsubseteq \overline{R'}$, it suffices to show that $\mu_1 \circ \T(R) \sqsubseteq \mu_1 \circ \T(R')$ (pointwise). To this end, we note that commutativity of the map $\delta$ with the monad multiplication, proved in \cite[Lemma~15\,(iii)]{CoumansJ2011} and captured by the commutativity of the lower diagram below (via the plain arrows)
\[\UseComputerModernTips\xymatrix{
\T^2 1 \ar[r]^-{\mu_1} & \T 1\\
\T^2(1 + 1) \ar[r]^-{\mu_{1 + 1}} \ar@<+1ex>[d]^-{\T \delta} \ar[u]^-{\T^2 !} & \T(1 + 1) \ar@<-1ex>[dd]_-{\delta} \ar[u]^-{\T !}\\
\T(\T 1 \times \T 1) \ar[d]^-{\langle \T \pi_1,\T \pi_2 \rangle} \ar@<+1ex>@{-->}[u]^-{\T q_{1,1}}\\
\T^2 1 \times \T^2 1 \ar[r]_-{\mu_1 \times \mu_1}  & \T 1 \times \T 1 \ar@<-1ex>@{-->}[uu]_-{q_{1,1}}
}\]
also yields commutativity of the whole diagram (via the dashed arrows). This formalises the commutativity of $+$ (defined as $\T ! \circ q_{1,1}$) with the monad multiplication. Now pre-composing this commutative diagram (dashed arrows) with the map
\[\UseComputerModernTips\xymatrix{
\T(X \times Y) \ar[r] & \T(\T 1 \times \T 1)
}\]
given by the image under $\T$ of the map $(x,y) \mapsto \langle R(x,y),S(x,y) \rangle$ yields
\[(\mu_1 \circ \T(R)) + (\mu_1 \circ \T(S)) = \mu_1 \circ \T(R+S) = \mu_1 \circ \T R'\]
and therefore, using the definition of $\sqsubseteq$, $\mu_1 \circ \T(R) \sqsubseteq \mu_1 \circ \T(R')$. This concludes the proof.
\end{proof}

Thus, $L_\T$ is a functor making the following diagram commute:
\[\UseComputerModernTips\xymatrix{
\TRel \ar[d]_-{q} \ar[r]^-{L_\T} & \TRel \ar[d]^-{q} \\
\Set \times \Set \ar[r]_-{\T \times \Id} & \Set \times \Set}\]

We are finally ready to give an alternative account of maximal traces of $\T \circ F$-coalgebras.

\begin{definition}
\label{max-trace-map}
Let $(C,\gamma)$ denote a $\T \circ F$-coalgebra, and let $(Z,\zeta)$ denote the final $F$-coalgebra. The \emph{maximal trace map $\tr_\gamma : C \to (\T 1)^Z$ of $\gamma$} is the exponential transpose of the greatest fixpoint $R : C \times Z \to \T1$ of the operator $\Op : \Rel_{C,Z} \to \Rel_{C,Z}$ given by the composition
\[\UseComputerModernTips\xymatrix@+0.3pc{
\Rel_{C,Z} \ar[r]^-{\Rel(F)} & \Rel_{F C,F Z} \ar[r]^-{L_\T} & \Rel_{\T(F C), F Z} \ar[r]^-{(\gamma \times \zeta)^*} & \Rel_{C,Z}
}\]
\end{definition}
The above definition appeals to the existence of least fixpoints in chain-complete partial orders, as formalised in the following fixpoint theorem from \cite{Priestley2002}.
\begin{theorem}[{\cite[8.22]{Priestley2002}}]
Let $P$ be a complete partial order and let $\Op : P \to P$ be order-preserving. Then $\Op$ has a least fixpoint.
\end{theorem}
Definition~\ref{max-trace-map} makes use of this result applied to the \emph{dual} of the order $\sqsubseteq$. Our assumption that $\sqsubseteq$ is $\omega^\op$-chain complete makes the dual order a complete partial order. Monotonicity of the operator in Definition~\ref{max-trace-map} is an immediate consequence of the functoriality of $\Rel(F)$, $L_\T$ and $(\gamma \times \delta)^*$.

\cite{Priestley2002} also gives a construction for the least fixpoint of an order-preserving operator on a complete partial order, which involves taking a limit over an ordinal-indexed chain. Instantiating this construction to the dual of the order $\sqsubseteq$ yields an ordinal-indexed sequence of relations $(R_\alpha)$, where:
\begin{itemize}
\item $R_0 = \top$ (i.e.~the relation on $C \times D$ given by $(c,d) \mapsto 1$),
\item $R_{\alpha+1} = \Op(R_\alpha)$,
\item $R_\alpha = \sqcap_{\beta < \alpha} R_{\beta}$, if $\alpha$ is a limit ordinal.
\end{itemize}
\begin{remark}
While in the case $\T = \Pow$, restricting to finite-state coalgebras $(C,\gamma)$ and $(D,\delta)$ results in the above sequence of relations stabilising in a finite number of steps, for $\T = \SDist$ or $T = \TW$ this is not in general the case. However, for probabilistic or weighted computations, an approximation of the greatest fixpoint may be sufficient for verification purposes, since a threshold can be provided as part of a verification task.
\end{remark}

\begin{remark}
By replacing the $F$-coalgebra $(Z,\zeta)$ by $(I,\alpha^{-1})$ with $(I,\alpha)$ an \emph{initial} $F$-algebra, one obtains an alternative account of \emph{finite} traces of states in $\T \circ F$-coalgebras, with the \emph{finite trace map} $\ftr_\gamma : C \to (\T 1)^I$ of a $\T\circ F$-coalgebra $(C,\gamma)$ being obtained via the greatest fixpoint of essentially the same operator $\Op$, but this time on $\Rel_{C,I}$. In fact, one can use any $F$-coalgebra in place of $(Z,\zeta)$, and for a specific verification task, a coalgebra with a finite state space, encoding a given linear-time behaviour, might be all that is required. 
\end{remark}

\begin{remark}
The choice of functor $F$ directly impacts on the notion of linear-time behaviour. For example, by regarding labelled transition systems as coalgebras of type $\Pow(A \times \Id)$ instead of $\Pow(1 + A \times \Id)$ (i.e.~not modelling successful termination explicitly), finite traces are not anymore accounted for -- the elements of the final $F$-coalgebra are given by infinite sequences of elements of $A$. This should not be regarded as a drawback, in fact it illustrates the flexibility of our approach.
\end{remark}

\begin{example}
Let $F$ denote an arbitrary polynomial functor (e.g.~$1 + A \times \Id$).
\begin{itemize}
\item For $T = \Pow$, the extension lifting $L_\Pow : \Rel \to \Rel$ takes a (standard) relation $R \subseteq X \times Y$ to the relation $L_\Pow(R) \subseteq \Pow (X) \times Y$ given by
\[(U,y) \in L_\Pow(R) \text{ ~if and only if~ there exists } x \in U \text{ with } (x,y) \in R\]
As a result, the greatest fixpoint of $\Op$ relates a state $c$ in a $\Pow \circ F$-coalgebra $(C,\gamma)$ with a state $z$ of the final $F$-coalgebra if and only if there exists a sequence of choices in the unfolding of $\gamma$ starting from $c$, that results in an $F$-behaviour bisimilar to $z$. This was made more precise in \cite{cirstea-11}, where infinite two-player games were developed for verifying whether a state of a $\Pow \circ F$-coalgebra has a certain maximal trace (element of the final $F$-coalgebra).
\item For $T = T_\SDist$, the extension lifting $L_\SDist : \Rel \to \Rel$ takes a valuation $R : X \times Y \to [0,1]$ to the valuation $L_\SDist(R) : \SDist(X) \times Y \to [0,1]$ given by
\[L_\SDist(R)(\varphi,y) = \sum\limits_{x \in \sup(\varphi)} \varphi(x) * R(x,y)\]
Thus, the greatest fixpoint of $\Op$ yields, for each state in a $\SDist \circ F$-coalgebra and each potential maximal trace $z$, the probability of this trace being exhibited. As computing these probabilities amounts to multiplying infinitely-many probability values, the probability of an infinite trace will often turn out to be $0$ (unless from some point in the unfolding of a particular state, probability values of $1$ are associated to the individual transitions that match a particular infinite trace). This may appear as a deficiency of our framework, and one could argue that a measure-theoretic approach, whereby a probability measure is derived from the probabilities of finite prefixes of infinite traces, would be more appropriate. Future work will investigate the need for a measure-theoretic approach. At this point, we simply point out that in a future extension of the present approach to linear-time logics (where individual maximal traces are to be replaced by linear-time temporal logic formulas), this deficiency is expected to disappear.
\item For $\T = \T_W$, the extension lifting $L_W : \Rel \to \Rel$ takes a \emph{weighted relation} $R : X \times Y \to W$ to the relation $L_W(R) : \T_W(X) \times Y \to W$ given by
\[L_W(R)(f,y) = \min_{x \in \sup(f)} (f(x) + R(x,y))\]
for $f : X \to W$ and $y \in Y$. Thus, the greatest fixpoint of $\Op$ maps a pair $(c,z)$, with $c$ a state in a $\T_W \circ F$-coalgebra and $z$ a maximal trace, to the \emph{cost} (computed via the $\min$ function) of exhibiting that trace. The case of weighted computations is somewhat different from our other two examples of branching types, in that the computation of the fixpoint starts from a relation that maps each pair of states $(c,z)$ to the value $0 \in \mathbb N^\infty$ (the top element for $\sqsubseteq$), and refines this down (w.r.t.~the $\sqsubseteq$ order) through stepwise unfolding of the coalgebra structures $\gamma$ and $\zeta$.
\end{itemize}
\end{example}
The approach presented above also applies to coalgebras of type $G \circ \T$ with $G$ a polynomial endofunctor, and more generally to coalgebras whose type is obtained as the composition of polynomial endofunctors and the monad $\T$, with possibly several occurrences of $\T$ in this composition. In the case of $G \circ \T$-coalgebras, instantiating our approach yields different results to the extension semantics proposed in \cite{JacobsSS12}. Specifically, the instantiation involves taking $(Z,\zeta)$ to be a final $G$-coalgebra and $(C,\gamma)$ to be an arbitrary $G \circ \T$-coalgebra, and considering the monotone operator on $\Rel_{C,Z}$ given by the composition
\begin{equation}
\label{o1}
\UseComputerModernTips\xymatrix@+0.3pc{
\Rel_{C,Z} \ar[r]^-{L_\T} & \Rel_{\T C,Z} \ar[r]^-{\Rel(G)} & \Rel_{G(\T C), G Z} \ar[r]^-{(\gamma \times \zeta)^*} & \Rel_{C,Z}
}
\end{equation}
The following example illustrates the difference between our approach and that of \cite{JacobsSS12}.
\begin{example}
For $G = 2 \times \Id^A$ with $A$ a finite alphabet and $\T = \Pow$, $G \circ \T$-coalgebras are non-deterministic automata, whereas the elements of the final $G$-coalgebra are given by functions $z : A^* \to 2$ and correspond to languages over $A$. In this case, the greatest fixpoint of the operator in (\ref{o1}) maps a pair $(c,z)$, with $c$ a state of the automaton and $z$ a language over $A$, to $\top$ if and only if there exists a sequence of choices in the unfolding of the automaton starting from $c$ that results in a deterministic automaton which accepts the language denoted by $z$. Taking the union over all $z$ such that $(c,z)$ is mapped to $\top$ now gives the language accepted by the non-deterministic automaton with $c$ as initial state, but only under the assumption that for each $a \in A$, an $a$-labelled transition exists from any state of the automaton. This example points to the need to further generalise our approach, so that in particular it can also be applied to pairs consisting of a $G \circ \T$-coalgebra and a $G'$-coalgebra, with $G'$ different from $G$. This would involve considering relation liftings for pairs of (polynomial) endofunctors. We conjecture that taking $G$ and $\T$ as above and $G' = 1 + A \times \Id$ would allow us to recover the notion of acceptance of a finite word over $A$ by a non-deterministic automaton.
\end{example}
Finally, we sketch the general case of coalgebras whose type is obtained as the composition of several endofunctors on $\Set$, one of which is a monad $\T$ that accounts for the presence of branching in the system, while the remaining endofunctors are polynomial and jointly determine the notion of linear-time behaviour. For simplicity of presentation, we only consider coalgebras of type $G \circ \T \circ F$, with the final $G \circ F$-coalgebra $(Z,\zeta)$ providing the domain of possible linear-time behaviours.
\begin{definition}
\label{linear-time-beh}
The \emph{linear-time behaviour} of a state in a coalgebra $(C,\gamma)$ of type $G \circ {\T} \circ F$ is the greatest fixpoint of an operator $\Op$ on $\Rel_{C,Z}$ defined by the composition:
\begin{equation}
\label{o2}
\UseComputerModernTips\xymatrix@+0.3pc{
\Rel_{C,Z} \ar[r]^-{\Rel(F)} & \Rel_{F C,F Z} \ar[r]^-{L_\T} & \Rel_{\T(F C), F Z}  \ar[r]^-{\Rel(G)} & \Rel_{G(\T F C), G F Z} \ar[r]^-{(\gamma \times \zeta)^*} & \Rel_{C,Z}
}
\end{equation}
\end{definition}
The greatest fixpoint of $\Op$ measures the extent with which a state in a $G \circ \T \circ F$-coalgebra can exhibit a given linear behaviour (element of the final $G \circ F$-coalgebra). Definition~\ref{linear-time-beh} generalises straightforwardly to coalgebraic types given by arbitrary compositions of polynomial endofunctors and the monad $\T$, with the extension lifting $L_\T$ being used once for each occurrence of $\T$ in such a composition. 
\begin{example}
\label{input-output}
Coalgebras of type $G \circ \T \circ F$, where $G = (1 + \Id)^A$ and $F = \Id \times B$,  model systems with branching, with both inputs (from a finite set $A$)  and outputs (in a set $B$). In this case, the possible linear behaviours are given by special trees, with both finite and infinite branches, whose edges are labelled by elements of $A$ (from each node, one outgoing edge for each $a \in A$), and whose nodes (with the exception of the root) are either labelled by $* \in 1$ (for leaves) or by an element of $B$ (for non-leaves). The linear-time behaviour of a state in a $G \circ \T \circ F$-coalgebra is then given by:
\begin{itemize}
\item the set of trees that can be exhibited from that state, when ${\T = \Pow}$,
\item the probability of exhibiting each tree (with the probabilities corresponding to different branches being \emph{multiplied} when computing this probability), when ${\T = \SDist}$, and 
\item the minimum cost of exhibiting each tree (with the costs of different branches being \emph{added} when computing this cost), when ${\T = \TW}$.
\end{itemize}
\end{example}
The precise connection between our approach and earlier work in \cite{HasuoJS07,cirstea-11,JacobsSS12} is yet to be explored. In particular, our assumptions are different from those of loc.\,cit., for example in \cite{HasuoJS07} the DCPO$_\bot$-enrichedness of the Kleisli category of $\T$ is required.

\begin{remark}
Our approach does not \emph{directly} apply to the probability distribution monad (defined similarly to the sub-probability distribution monad, but with probabilities adding up to exactly $1$), as this monad does not satisfy the condition $\T \emptyset = 1$ of Definition~\ref{additive}. However, systems where branching is described using probability distributions can still be dealt with, by regarding all probability distributions as sub-probability distributions.
\end{remark}

In the remainder of this section, we briefly explore the usefulness of an operator similar to $\Op$, which employs a similar extension lifting arising from the \emph{double strength} of the monad $\T$. We begin by noting that a result similar to Proposition~\ref{prop-kock} is proved in \cite{Kock12} for a commutative monad on a cartesian closed category.
\begin{proposition}[{\cite[Proposition~9.3]{Kock12}}]
Let $(B,\beta)$ be a $\T$-algebra. Then any $f : X \times Y \to B$ extends uniquely along $\eta_X \times \eta_Y$ to a bilinear $\tilde{f} : \T X \times \T Y \to B$, making the following triangle commute:
\[\UseComputerModernTips\xymatrix{
\T X \times \T Y \ar[r]^-{\tilde{f}} & B \\
X \times Y \ar[u]^-{\eta_X \times \eta_Y} \ar[ur]_-{f}
}\]
\end{proposition}
Here, bilinearity amounts to linearity in each argument.
\begin{definition}
For a commutative monad $\T : \Set \to \Set$, the \emph{double extension lifting} $L_\T' : \Rel \to \Rel$ is the functor taking a relation $R : X \times Y \to \T1$ to its unique bilinear extension $\tilde{R} : \T X \times \T Y \to \T1$.
\end{definition}
\begin{remark}
\label{alt-lifting}
An alternative definition of $L_\T'$ is as the composition of $L_\T$ with a dual lifting, which takes a relation $R : X \times Y \to \T 1$ to its unique $2$-linear extension $\overline{R} : X \times \T Y \to \T 1$.
\end{remark}
\begin{remark}
Again, it can be shown that a direct definition of the relation $\tilde{R} : \T X \times \T Y \to \T 1$ is as the composition
\[\UseComputerModernTips\xymatrix@+0.3pc{
\T X \times \T Y \ar[r]^-{\dst_{X,Y}} & \T(X \times Y) \ar[r]^-{\T (R)} & \T^2 1 \ar[r]^-{\mu_1} & \T 1
}\]
\end{remark}

\begin{proposition}
\label{prop-functoriality-dual}
The mapping $R \in \Rel_{X,Y} \mapsto \overline{R} \in \Rel_{X,\T Y}$ is functorial.
\end{proposition}
We now fix \emph{two} $\T \circ F$-coalgebras $(C,\gamma)$ and $(D,\delta)$ and explore the greatest fixpoint of the operator $\Op' : \Rel_{C,D} \to \Rel_{C,D}$ defined by the composition
\[\UseComputerModernTips\xymatrix@+0.3pc{
\Rel_{C,D} \ar[r]^-{\Rel(F)} & \Rel_{F C,F D} \ar[r]^-{L_\T'} & \Rel_{\T(F C), \T(F D)} \ar[r]^-{(\gamma \times \zeta)^*} & \Rel_{C,D}
}\]

As before, the operator $\Op'$ is monotone and therefore admits a greatest fixpoint. We argue that this fixpoint also yields useful information regarding the linear-time behaviour of states in $\T \circ F$-coalgebras. Moreover, this generalises to coalgebras whose types are arbitrary compositions of polynomial functors and the branching monad $\T$. This is expected to be of relevance when extending the linear-time view presented here to linear-time logics and associated formal verification techniques. The connection to formal verification constitutes work in progress, but the following examples motivate our claim that the lifting $L_\T'$ is worth further exploration.

\begin{example}Let $F : \Set \to \Set$ be a polynomial endofunctor, describing some linear-type behaviour.
\begin{enumerate}
\item For non-deterministic systems (i.e.~$\Pow \circ F$-coalgebras), the greatest fixpoint of $\Op'$ relates two states if and only if they admit a common maximal trace.
\item For probabilistic systems (i.e.~$\SDist \circ F$-coalgebras), the greatest fixpoint of $\Op'$ measures the probability of two states exhibiting the same maximal trace.

\item For weighted systems (i.e.~$\T_W \circ F$-coalgebras), the greatest fixpoint of $\Op'$ measures the \emph{joint} minimal cost of two states exhibiting the same maximal trace. To see this, note that the lifting $L_W' : \Rel \to \Rel$ takes a weighted relation $R : X \times Y \to W$ to the relation $L_W'(R) : \T_W(X) \times \T_W(Y) \to W$ given by
\[L_W'(R)(f,g) = \min_{x \in \sup(f),y\in \sup(g)} (f(x) + g(y) + R(x,y))\]

\end{enumerate}
\end{example}

\section{Conclusions and Future Work}

We have provided a general and uniform account of the linear-time behaviour of a state in a coalgebra whose type incorporates some notion of branching (captured by a monad on $\Set$). Our approach is compositional, and so far applies to notions of linear behaviour specified by \emph{polynomial} endofunctors on $\Set$. The key ingredient of our approach is the notion of extension lifting, which allows the branching behaviour of a state to be abstracted away in a coinductive fashion.

Immediate future work will attempt to exploit the results of \cite{Ghani2011,Ghani2012} in order to define generalised relation liftings for \emph{arbitrary} endofunctors on $\Set$, and to extend our approach to other base categories. The work in loc.\,cit.~could also provide an alternative description for the greatest fixpoint used in Definition~\ref{linear-time-beh}.

The present work constitutes a stepping stone towards a coalgebraic approach to the formal verification of linear-time properties. This will employ linear-time coalgebraic temporal logics for the specification of system properties, and automata-based techniques for the verification of these properties, as outlined in \cite{Cirstea11} for the case of non-deterministic systems.

\bibliography{fics}

\begin{thebibliography}{10}
\providecommand{\bibitemdeclare}[2]{}
\providecommand{\surnamestart}{}
\providecommand{\surnameend}{}
\providecommand{\urlprefix}{Available at }
\providecommand{\url}[1]{\texttt{#1}}
\providecommand{\href}[2]{\texttt{#2}}
\providecommand{\urlalt}[2]{\href{#1}{#2}}
\providecommand{\doi}[1]{doi:\urlalt{http://dx.doi.org/#1}{#1}}
\providecommand{\bibinfo}[2]{#2}

\bibitemdeclare{article}{cirstea-11}
\bibitem{cirstea-11}
\bibinfo{author}{Corina \surnamestart C\^{\i}rstea\surnameend}
  (\bibinfo{year}{2011}): \emph{\bibinfo{title}{Maximal Traces and Path-Based
  Coalgebraic Temporal Logics}}.
\newblock {\sl \bibinfo{journal}{Theoretical Computer Science}}
  \bibinfo{volume}{412}(\bibinfo{number}{38}), pp. \bibinfo{pages}{5025--5042},
  \doi{10.1016/j.tcs.2011.04.025}.

\bibitemdeclare{inproceedings}{Cirstea11}
\bibitem{Cirstea11}
\bibinfo{author}{Corina \surnamestart C\^{\i}rstea\surnameend}
  (\bibinfo{year}{2011}): \emph{\bibinfo{title}{Model Checking Linear
  Coalgebraic Temporal Logics: An Automata-Theoretic Approach}}.
\newblock In: {\sl \bibinfo{booktitle}{Proc.\,\,CALCO\,\,2011}}, {\sl
  \bibinfo{series}{Lecture Notes in Computer Science}} \bibinfo{volume}{6859},
  \bibinfo{publisher}{Springer}, pp. \bibinfo{pages}{130--144},
  \doi{10.1007/978-3-642-22944-2\_10}.

\bibitemdeclare{incollection}{CoumansJ2011}
\bibitem{CoumansJ2011}
\bibinfo{author}{Dion \surnamestart Coumans\surnameend} \&
  \bibinfo{author}{Bart \surnamestart Jacobs\surnameend}
  (\bibinfo{year}{2013}): \emph{\bibinfo{title}{Scalars, Monads, and
  Categories}}.
\newblock In \bibinfo{editor}{C.~\surnamestart Heunen\surnameend},
  \bibinfo{editor}{M.~\surnamestart Sadrzadeh\surnameend} \&
  \bibinfo{editor}{E.~\surnamestart Grefenstette\surnameend}, editors: {\sl
  \bibinfo{booktitle}{Quantum Physics and Linguistics. A Compositional,
  Diagrammatic Discourse}}, \bibinfo{publisher}{Oxford Univ. Press}, pp.
  \bibinfo{pages}{184--216}, \doi{10.1093/acprof:oso/9780199646296.001.0001}.

\bibitemdeclare{book}{Priestley2002}
\bibitem{Priestley2002}
\bibinfo{author}{Brian~A. \surnamestart Davey\surnameend} \&
  \bibinfo{author}{Hilary~A. \surnamestart Priestley\surnameend}
  (\bibinfo{year}{2002}): \emph{\bibinfo{title}{Introduction to Lattices and
  Order (2. ed.)}}.
\newblock \bibinfo{publisher}{Cambridge University Press},
  \doi{10.1017/CBO9780511809088}.

\bibitemdeclare{misc}{EsikK07}
\bibitem{EsikK07}
\bibinfo{author}{Zoltan \surnamestart \'Esik\surnameend} \&
  \bibinfo{author}{Werner \surnamestart Kuich\surnameend}
  (\bibinfo{year}{2007}): \emph{\bibinfo{title}{Modern Automata Theory}}.
\newblock \bibinfo{note}{{http://dmg.tuwien.ac.at/kuich/}}.

\bibitemdeclare{inproceedings}{Ghani2011}
\bibitem{Ghani2011}
\bibinfo{author}{Cl{\'e}ment \surnamestart Fumex\surnameend},
  \bibinfo{author}{Neil \surnamestart Ghani\surnameend} \&
  \bibinfo{author}{Patricia \surnamestart Johann\surnameend}
  (\bibinfo{year}{2011}): \emph{\bibinfo{title}{Indexed Induction and
  Coinduction, Fibrationally}}.
\newblock In: {\sl \bibinfo{booktitle}{Proc.\,\,CALCO\,\,2011}}, {\sl
  \bibinfo{series}{Lecture Notes in Computer Science}} \bibinfo{volume}{6859},
  \bibinfo{publisher}{Springer}, pp. \bibinfo{pages}{176--191},
  \doi{10.1007/978-3-642-22944-2\_13}.

\bibitemdeclare{article}{Ghani2012}
\bibitem{Ghani2012}
\bibinfo{author}{Neil \surnamestart Ghani\surnameend},
  \bibinfo{author}{Patricia \surnamestart Johann\surnameend} \&
  \bibinfo{author}{Cl{\'e}ment \surnamestart Fumex\surnameend}
  (\bibinfo{year}{2012}): \emph{\bibinfo{title}{Generic Fibrational
  Induction}}.
\newblock {\sl \bibinfo{journal}{Logical Methods in Computer Science}}
  \bibinfo{volume}{8}(\bibinfo{number}{2}), \doi{10.2168/LMCS-8(2:12)2012}.

\bibitemdeclare{inproceedings}{Hasuo12}
\bibitem{Hasuo12}
\bibinfo{author}{Ichiro \surnamestart Hasuo\surnameend}, \bibinfo{author}{Kenta
  \surnamestart Cho\surnameend}, \bibinfo{author}{Toshiki \surnamestart
  Kataoka\surnameend} \& \bibinfo{author}{Bart \surnamestart Jacobs\surnameend}
  (\bibinfo{year}{2013}): \emph{\bibinfo{title}{Coinductive Predicates and
  Final Sequences in a Fibration}}.
\newblock In: {\sl \bibinfo{booktitle}{Proc.\,\,MFPS\,\,XXIX}}, pp.
  \bibinfo{pages}{181--216}.

\bibitemdeclare{article}{HasuoJS07}
\bibitem{HasuoJS07}
\bibinfo{author}{Ichiro \surnamestart Hasuo\surnameend}, \bibinfo{author}{Bart
  \surnamestart Jacobs\surnameend} \& \bibinfo{author}{Ana \surnamestart
  Sokolova\surnameend} (\bibinfo{year}{2007}): \emph{\bibinfo{title}{Generic
  Trace Semantics via Coinduction}}.
\newblock {\sl \bibinfo{journal}{Logical Methods in Computer Science}}
  \bibinfo{volume}{3}(\bibinfo{number}{4}), pp. \bibinfo{pages}{1--36},
  \doi{10.2168/LMCS-3(4:11)2007}.

\bibitemdeclare{article}{HermidaJ98}
\bibitem{HermidaJ98}
\bibinfo{author}{Claudio \surnamestart Hermida\surnameend} \&
  \bibinfo{author}{Bart \surnamestart Jacobs\surnameend}
  (\bibinfo{year}{1998}): \emph{\bibinfo{title}{Structural Induction and
  Coinduction in a Fibrational Setting}}.
\newblock {\sl \bibinfo{journal}{Inf. Comput.}}
  \bibinfo{volume}{145}(\bibinfo{number}{2}), pp. \bibinfo{pages}{107--152},
  \doi{10.1006/inco.1998.2725}.

\bibitemdeclare{misc}{JacobsBook}
\bibitem{JacobsBook}
\bibinfo{author}{Bart \surnamestart Jacobs\surnameend} (\bibinfo{year}{2012}):
  \emph{\bibinfo{title}{Introduction to Coalgebra. Towards Mathematics of
  States and Observations (Version 2.0)}}.
\newblock \bibinfo{note}{Draft}.

\bibitemdeclare{inproceedings}{JacobsSS12}
\bibitem{JacobsSS12}
\bibinfo{author}{Bart \surnamestart Jacobs\surnameend},
  \bibinfo{author}{Alexandra \surnamestart Silva\surnameend} \&
  \bibinfo{author}{Ana \surnamestart Sokolova\surnameend}
  (\bibinfo{year}{2012}): \emph{\bibinfo{title}{Trace Semantics via
  Determinization}}.
\newblock In: {\sl \bibinfo{booktitle}{Proc.\,\,CMCS\,\,2012}}, {\sl
  \bibinfo{series}{Lecture Notes in Computer Science}} \bibinfo{volume}{7399},
  \bibinfo{publisher}{Springer}, pp. \bibinfo{pages}{109--129},
  \doi{10.1007/978-3-642-32784-1}.

\bibitemdeclare{article}{KS90}
\bibitem{KS90}
\bibinfo{author}{Paris~C. \surnamestart Kanellakis\surnameend} \&
  \bibinfo{author}{Scott~A. \surnamestart Smolka\surnameend}
  (\bibinfo{year}{1990}): \emph{\bibinfo{title}{CCS Expressions, Finite State
  Processes, and Three Problems of Equivalence}}.
\newblock {\sl \bibinfo{journal}{Inf. Comput.}}
  \bibinfo{volume}{86}(\bibinfo{number}{1}), pp. \bibinfo{pages}{43--68},
  \doi{10.1016/0890-5401(90)90025-D}.

\bibitemdeclare{misc}{Kock2011}
\bibitem{Kock2011}
\bibinfo{author}{Anders \surnamestart Kock\surnameend} (\bibinfo{year}{2011}):
  \emph{\bibinfo{title}{Monads and extensive quantities}}.
\newblock \bibinfo{note}{ArXiv:1103.6009}.

\bibitemdeclare{article}{Kock12}
\bibitem{Kock12}
\bibinfo{author}{Anders \surnamestart Kock\surnameend} (\bibinfo{year}{2012}):
  \emph{\bibinfo{title}{Commutative monads as a theory of distributions}}.
\newblock {\sl \bibinfo{journal}{Theory and Applications of Categories}}
  \bibinfo{volume}{26}(\bibinfo{number}{4}), pp. \bibinfo{pages}{97--131}.

\bibitemdeclare{article}{Moss99}
\bibitem{Moss99}
\bibinfo{author}{Lawrence~S. \surnamestart Moss\surnameend}
  (\bibinfo{year}{1999}): \emph{\bibinfo{title}{Coalgebraic Logic}}.
\newblock {\sl \bibinfo{journal}{Ann. Pure Appl. Logic}}
  \bibinfo{volume}{96}(\bibinfo{number}{1-3}), pp. \bibinfo{pages}{277--317},
  \doi{10.1016/S0168-0072(98)00042-6}.

\bibitemdeclare{article}{Pattinson03}
\bibitem{Pattinson03}
\bibinfo{author}{Dirk \surnamestart Pattinson\surnameend}
  (\bibinfo{year}{2003}): \emph{\bibinfo{title}{Coalgebraic modal logic:
  soundness, completeness and decidability of local consequence}}.
\newblock {\sl \bibinfo{journal}{Theor. Comput. Sci.}}
  \bibinfo{volume}{309}(\bibinfo{number}{1-3}), pp. \bibinfo{pages}{177--193},
  \doi{10.1016/S0304-3975(03)00201-9}.

\end{thebibliography}

\end{document}